%% file: FinalPaper.tex
\documentclass[conference, letter]{IEEEtran}

\usepackage[utf8]{inputenc}
\usepackage[english]{babel}
\usepackage{amsmath, amssymb, mathrsfs}
\usepackage{amsthm}
\usepackage{url}
\usepackage{listings}

\usepackage{graphicx}
\usepackage{color}
\usepackage{comment}
\usepackage[boxed,linesnumbered]{algorithm2e}

\newcommand{\N}{\ensuremath{\mathbb {N}}}

\renewcommand{\epsilon}{\varepsilon}

\graphicspath{{./images/}}

\newtheorem{lem}{Lemma}

\begin{document}

\title{\sc Radix Conversion for IEEE754-2008 Mixed Radix Floating-Point Arithmetic}
\author{  
  \IEEEauthorblockN{Olga Kupriianova}
    \IEEEauthorblockA{ 
    UPMC Paris 6 -- LIP6 -- PEQUAN team\\
    4, place Jussieu\\
    75252 Paris Cedex 05, France\\
    Email: \url{olga.kupriianova@lip6.fr}}
  \and
  \IEEEauthorblockN{Christoph Lauter}
    \IEEEauthorblockA{
      UPMC Paris 6 -- LIP6 -- PEQUAN team\\ 
      4, place Jussieu\\
      75252 Paris Cedex 05, France\\    
      Email: \url{christoph.lauter@lip6.fr}}
  \and
  \IEEEauthorblockN{Jean-Michel Muller}
    \IEEEauthorblockA{
    CNRS -- ENS Lyon -- Université de Lyon\\
    46, allée d'Italie\\
    69364 Lyon Cedex 07, France\\
    Email: \url{jean-michel.muller@ens-lyon.fr}}
    }
\maketitle

\begin{abstract}
  Conversion between binary and decimal floating-point representations
  is ubiquitous. Floating-point radix conversion means converting both
  the exponent and the mantissa. We develop an atomic
  operation for FP radix conversion with simple straight-line
  algorithm, suitable for hardware design. Exponent conversion is
  performed with a small multiplication and a lookup table. It yields
  the correct result without error. Mantissa conversion uses a few
  multiplications and a small lookup table that is shared amongst all
  types of conversions. The accuracy changes by adjusting the computing
  precision.

\end{abstract}

\section{Introduction}

Humans are used to operate decimals while almost all the hardware is
binary. According to IEEE754-2008 norm~\cite{IEEE754-2008} a floating
point number is represented as $\beta^E\cdot m$, where $\beta^{p-1}
\le m \le \beta^p - 1$; $p$ is precision, $m \in \mathbb{N}$ is
mantissa, $E \in \mathbb{Z}$ is exponent and $\beta$, the base or
radix, is either two or ten. When the base $\beta=2$, we have binary
floating point (FP) numbers, when $\beta = 10$, the decimal
one. However, most of hardware is binary, so the decimal mantissas are
actually coded in binary. The formats for both radices differ by the
length of stored numbers. Standartization of decimal FP arithmetic
brings new challenges, e.g. supporting decimal transcendental
functions with essentially binary
hardware~\cite{Harrison2009}. In~\cite{Harrison2009} in order to
evaluate decimal transcendental function the format conversion is used
twice. The IEEE standard requires~\cite{IEEE754-2008} the
implementation of all the operations for different formats, but only
for the operands of the same radix. The format does not require any
mixed radix operations, i.e. one of the operands is binary, the other
is decimal.  Mixed radix arithmetic is currently being developed,
although there are already some approaches
published~\cite{CHATSG2009},~\cite{BLMM2013}.

Floating point radix conversion (from binary to decimal and vice
versa) is a widespread operation, the simplest examples are the
\texttt{scanf} and \texttt{printf} functions. It could also exist as
an operation for financial applications or as a ``precomputing step''
for mixed radix operations. The radix conversion is used in number
conversion operations, and implicitly in \texttt{scanf} and
\texttt{printf} operations.

The current implementations of \texttt{scanf} and \texttt{printf} are
correct only for one rounding mode and allocate a lot of memory. In
this paper we develop a unified atomic operation for the conversion,
so all the computations can be done in integer with the precomputed
memory consumption.

While radix conversion is a very common operation, it comes in
different variants that are mostly coded in ad-hoc way in existing
code. However, radix conversion always breaks down into to elementary
steps: determining an exponent of the output radix and computing a
mantissa in the output radix. Section~\ref{sec:steps} describes the
2-steps approach of the radix conversion, section~\ref{sec:Exponent}
contains the algorithm for the exponent computation,
section~\ref{sec:Mantissa} presents a novel approach of raising $5$ to
an integer power used in the second step of the radix-conversion that
computes the mantissa. Section~\ref{sec:error analysis} contains
accuracy bounds for the algorithm of raising five to a huge power,
section~\ref{sec:impldet} describes some implementation tricks and
presents experimental results.

\section{Two-steps Radix Conversion Algorithm}
\label{sec:steps}
Conversion from a binary FP representation $2^E \cdot m$, where $E$ is
the binary exponent and $m$ is the mantissa, to a decimal
representation $10^F \cdot n$, requires two steps: determination of
the decimal exponent $F$ and computation of the mantissa $n$. The
conversion back to binary is pretty similar except of an extra step
that will be explained later. Here and after consider the normalized
mantissas $n$ and $m$: $10^{p_{10}-1} \le n \le 10^{p_{10}} - 1$
and $2^{p_{2}-1} \le m \le 2^{p_{2}} - 1$, where $p_{10}$ and $p_2$
are the decimal and binary precisions respectively. The exponents $F$
and $E$ are bounded by some values depending on the IEEE754-2008 format.

In order to enclose the converted decimal mantissa $n$ into one
decade, for a certain output precision $p_{10}$, the decimal exponent
$F$ has to be computed~\cite{MullerEtAl2010} as follows:
\begin{equation}
  \label{eq:decimalExp}
  F = \left\lfloor \log_{10}(2^E \cdot m) \right\rfloor - p_{10} + 1.
\end{equation}
The most difficult thing here is the evaluation of the logarithm: as
the function is transcendental, the result is always an approximation
and function call is extremely expensive. Present algorithm computes
the exponent (\ref{eq:decimalExp}) for a new-radix floating-point
number only with a multiplication, binary shift, a precomputed
constant and a lookup table (see section~\ref{sec:Exponent}).

Once $F$ is determined, the mantissa $n$ is given as
\begin{equation}
\label{eq:decimalMant}
n = *_{p_{10}} \left( \frac{2^E \cdot m}{10^F}\right),
\end{equation}
where $*_{p_{10}}$ corresponds to the current rounding mode (to the
nearest, rounding down, or rounding up~\cite{IEEE754-2008}).  The
conversions are always done with some error $\varepsilon$, so the
following relation is fulfilled: $10^F \cdot n = 2^E \cdot m \cdot (1
+ \varepsilon)$. In order to design a unique algorithm for all the
rounding modes it is useful to compute $n^*$, such that $10^F \cdot
n^* = 2^E \cdot m$. 
Thus, we get the following expression for the decimal mantissa:
$$n^* = 2^{E- F} 5^{-F}m$$
As $2^{E - F}$ is a simple binary shift and the multiplication by $m$
is small, the binary-to-decimal mantissa conversion reduces to
compute the leading bits of $5^{-F}$. 

The proposed ideas apply with minor changes to decimal-to-binary
conversion: the base of the logarithm is 2 on the exponent computation
step and one additional step is needed; for the mantissa
computation the power $5^F$ is required instead of $5^{-F}$.

\section{Loop-Less Exponent Determination}
\label{sec:Exponent}
The current implementations of the logarithm function are expensive
and produce approximated values. However, some earlier conversion
approaches computed this approximation \cite{Gay90} by Taylor series
or using iterations \cite{SW1990, BD1996}. Here the exponent for the
both conversions is computed exactly neither with \texttt{libm}
function call nor any polynomial approximation.

After performing one transformation step,~(\ref{eq:decimalExp}) can be rewritten as following:
$$ F = \left\lfloor E \log_{10}(2) + \left\lfloor\log_{10}(m) \right\rfloor + \left\{ \log_{10}(m) \right\}\right\rfloor -p_{10} + 1, $$
where $\{x\} = x - \lfloor x \rfloor$, the fractional part of the
number $x$.

As the binary mantissa $m$ is normalized in one binade $2^{p_{2}-1}
\le m < 2^{p_{2}} $, we can assume that it lies entirely in one
decade. If it is not the case, we can always scale it a little
bit. The inclusion in one decade means that
$\lfloor\log_{10}(m)\rfloor$ stays the same on the whole interval. So,
for the given format one can precompute and store this value as a
constant. Thus, it is possible to take the integer number
$\left\lfloor \log_{10}(m)\right\rfloor$ out of the floor
operation in the previous equation. After representing the first summand as a sum of it's
integer and fractional parts, we have the following expression under
the floor operation:
$$ \left\lfloor \left\lfloor E \log_{10}(2)\right\rfloor + \left\{E \log_{10}(2)\right\} + \left\{ \log_{10}(m) \right\}\right\rfloor. $$
Here we add two fractional parts to an integer. We add something that
is strictly less than two, so under the floor operation we have either
an integer plus some small fraction that will be thrown away, or an
integer plus one plus small fraction. Thus, we can take the fractional
parts out of the floor brackets adding a correction $\gamma$:
$$\left\lfloor E \log_{10}(2)\right\rfloor + \gamma,  \: \gamma \in \{0, 1\}.$$
This correction $\gamma$ equals to 1 when the sum of two fractional
parts from the previous expression exceeds 1, or mathematically:
$$  E\log_{10}(2) - \left\lfloor E\log_{10}(2) \right\rfloor  + \log_{10}(m) - \left\lfloor \log_{10}(m)\right\rfloor  \ge 1.$$
Due to the logarithm function the expression on the left is strictly
monotonous (increasing). This means that we need only one threshold
value $m^*(E)$, such that $\forall m \ge m^*(E)$ the correction
$\gamma=1$.  As we know the range for the exponents $E$ beforehand, we
can store the critical values $m^*(E)= 10^{1 - \left ( E \log_{10}2 -
    \left\lfloor E\log_{10}2\right\rfloor \right) + \left\lfloor
    \log_{10}(m)\right\rfloor}$ in a table.

There is a technique proposed in \cite{BM2008} to compute
$\left\lfloor E \log_{10}(2)\right\rfloor$ with a multiplication,
binary shift and the use of a precomputed constant. So, finally the
value of the decimal exponent can be obtained as
\begin{equation}
\label{eq:op_1^10Final}
F = \left\lfloor E \left\lfloor \log_{10}(2) \cdot 2^{\lambda} \right\rfloor \cdot 2^{-\lambda}\right\rfloor + \left\lfloor \log_{10}(m)\right\rfloor -p_{10} + 1 + \gamma 
\end{equation}
The algorithm pseudocode is provided below.

\begin{algorithm}[t]
  \label{algo:exp}
  \SetKwInOut{Input}{input}
  \SetKwInOut{Output}{output}
  \Input{$E$, $m$}
  $F \leftarrow E \cdot \lfloor \log_{10}(2) \cdot 2^{\lambda} \rfloor$; //multiply by a constant\;
  $F \leftarrow \lfloor F \cdot 2^{-\lambda}\rfloor$; //binary right shift\;
  $F \leftarrow  F + \lfloor \log_{10}(m) \rfloor + 1 - p_{10}$; //add a constant\;
  \If{$ m \ge m^*(E)$}{
      $F \leftarrow F + 1$\;
    }

  \caption{The exponent computation in the conversion from binary to decimal floating-point number}
\end{algorithm}

The decimal-to-binary conversion algorithm is the same with a small
additional remark. We want to convert decimal FP numbers to binary, so
the input mantissas are in the range $10^{p_{10}-1} \le n <
10^{p_{10}}$. As we mentioned, on this step the base of the logarithm
is 2, and the problem here is that $\lfloor\log_2(10)\rfloor=3$, so it
seems that we need three tables, but once we represent the decimal
mantissa $n$ as a binary FP number $n = 2^{\hat{E}} \hat{m}$ in some
precision $\kappa$, it suffices just one table. For all the possible
values $\hat{m}$ the following holds
$\lfloor\log_2(\hat{m})\rfloor=\kappa-1$. This mantissa representation
can be made exact: we'll have to shift the decimal $n$ to the
left. Thus, the precision of this new number is $\kappa = \lceil
\log_{2}(10^{p_{10}}-1) \rceil$.

So, the proposed algorithm works for both conversion
directions. However, one can notice, that for binary-to-decimal
conversion the table size can be even reduced by the factor of two. We
have used the mantissas from one binade: $2^{p_{2}-1} \le m <
2^{p_{2}} $. The whole reasoning stays the same if we scale these
bounds in order to have $1 \le m < 2$, the table entries quantity
stays the same. Now it is clear that
$\left\lfloor\log_{10}(m)\right\rfloor = 0$ for all these
mantissas. However, it still stays zero if we slightly modify the
mantissa's bounds: $\forall m': 1 \le m' < 4 ,\, \log_{10}(m') = 0$.
Thus, we get a new binary representation of the input: $2^{E'}m' =
2^{E}m$, where $E' = E - (E\mod2)$ and $m' = m \cdot 2^{E \mod
  2}$. So, we see that for the new mantissas interval we do not take
into account the last exponent bit. So, the table entries quantity for
the values $m^*(E)$ reduces twice. The corresponding interval for
mantissas is $[1, 4)$, because in this case we need to find the
remainders of two, that is just a binary shift. The interval $[1, 8)$
is larger, so it could reduce the table size even more, but requires
computation of the remainders of three.

The table sizes for some particular formats are small enough to be
integrated in hardware. However, these tables are quite multipurpose,
they are shared between all I/O and arithmetic decimal-binary FP
conversions, so, once they are coded, they could be used in all the
mixed radix operations. The corresponding table sizes for different
formats are listed in table~\ref{tab:mantissaLUT}.

\begin{table}[t]
\normalsize
\begin{center}
\begin{tabular}{|c|l|}
\hline
Initial Format & Table size \\
\hline
\texttt{binary32}       & 554 bytes\\    \hline
\texttt{binary64}       & 8392 bytes \\  \hline
\texttt{binary128}      & 263024 bytes \\ \hline

\texttt{decimal32}      & 792 bytes \\ \hline
\texttt{decimal64}      & 6294 bytes \\ \hline
\texttt{decimal128}     & 19713 bytes \\ \hline
\end{tabular}
\end{center}
\caption{Table size for exponent computation step}
\label{tab:mantissaLUT}
\end{table}

\section{Computing the Mantissa with the Right Accuracy}
\label{sec:Mantissa}

As it was mentioned, the problem on the second step is the computation
of the value $5^B$ with some bounded exponent $B \in \N$. If the
initial range for the exponent of five contains negative values, we
compute $5^{B + \bar{B}}$, where $\bar{B}$ is chosen in order to make
the range for the exponents nonnegative. In this case we store the
leading bits of $5^{-\bar{B}}$ as a constant and after computing $5^{B
  + \bar{B}}$ with the proposed algorithm, we multiply the result by
the constant.

In this section we propose an algorithm for raising five to a huge
natural power without rational arithmetic or divisions. The range for
these natural exponents $B$ is determined by the input format,
e.g. for the conversion from binary64 the range is about six hundred.

We propose to perform several Euclidean divisions in order to
represent the number $B$ the following way:
\begin{equation}
\label{eq:Bform}
B = 2^{n_k}\cdot q_k + 2^{n_{k-1}}q_{k-1} + \ldots + 2^{n_1} q_1 + q_0, 
\end{equation}
where $0 \le q_0 \le 2^{n_1}-1$, $n_k \ge n_{k-1}, \, k \ge 1$. The
mentioned divisions are just a chain of binary shifts. All the quotients
are in the same range and we assume that the range for $q_0$ is the largest
one, so we have $q_i \in [0; 2^{n_1}-1], \, 0 \le i \le k$.  Once the
exponent is represented as~(\ref{eq:Bform}), computation $5^B$ is done
with the respect to the following expression:
\begin{equation}
\label{eq:5Bform}
5^B = (5^{q_k})^{2^{n_k}}\cdot (5^{q_{k-1}})^{2^{n_{k-1}}} \cdot \ldots \cdot (5^{q_1})^{2^{n_1}} \cdot 5^{q_0}
\end{equation}
Let us analyze how the proposed formula can simplify the algorithm of
raising five to the power $B$. We mentioned that all the quotients
$q_i$ are bounded. By selecting the parameters $k$ and $n_i$ we can
make these quotients small, so the values $5^{q_i}$ can be stored in a
table. Then, each factor in~(\ref{eq:5Bform}) is a table value raised
to the power $2^{n_i}$ which is the same as a table value squared
$n_i$ times. 

So, the algorithm is the following: represent $B$ as (\ref{eq:Bform})
and get the values $q_i$, then for each $q_i$ get the table value
$5^{q_i}$ and perform $n_i$ squarings, and finally multiply all the
squared values beginning from the largest one. The scheme can be found
on Fig.~\ref{fig:algo2}, the pseudocode for squarings is in
algorithm~\ref{algo:squarings} and for the final multiplication step
in algorithm~\ref{algo:mult}. All these steps are done in order to
convert the FP numbers, so we simulate usual floating-point
computations in integer. The exponent $B$ is huge, the value $5^B$ is
also huge, so we can store only the leading bits. Thus, on each
multiplication step (squarings are also multiplications) we throw away
the last $\lambda$ bits. Of course these manipulations yield to an
error, in section~\ref{sec:error analysis} there are details and
proofs for the error analysis. 

There is still one detail in algorithm~\ref{algo:squarings} that
was not explained: the correction $\sigma_j$. The mantissa of the
input number is represented as a binary number bounded by one binade
(for both, binary and decimal formats). Assume that we operate the
numbers in the range $[2^{p-1}, 2^p)$. After each squaring we can get
a value less then infimum of this range. So, if the first bit of the
intermediate result after some squarings is 0, we shift it to the
left.

\begin{algorithm}[t]
\label{algo:squarings}
  \SetKwInOut{Input}{input}
  \Input{$ n_j, \, v_j = 5^{q_j}$}
  $\sigma_j \leftarrow 0$\;
  \For{$i \leftarrow 1$ \KwTo $n_j$} {
    $v_j\leftarrow \lfloor v_j^2 \cdot 2^{-\lambda} \rfloor $\;  
    shiftNeeded $\leftarrow 1 - \lfloor v_j \cdot 2^{1-p} \rfloor $ //get the first bit\;
    $v_j \leftarrow v_j \ll $ shiftNeeded\;
    $\sigma_j \leftarrow 2\cdot \sigma_j + $shiftNeeded\;
  }             
  result $\leftarrow v_j \cdot 2^{-\sigma_j} \cdot 2^{(2^{n_j}-1)\lambda}$\;                                              
\caption{Squaring with shifting $\lambda$ last bits}
\end{algorithm}  

The described algorithm is applied $k$ times to each factors
in~(\ref{eq:5Bform}). Then the last step is to multiply all the
factors starting from the largest power like in listing below.

  \begin{algorithm}[t]
    \label{algo:mult}
                $m \leftarrow $ 1\;
                \For{$i \leftarrow k$ \KwTo $1$}{
                  $m \leftarrow \left\lfloor (m \cdot v_i) \cdot 2^{-\lambda} \right\rfloor$\; 
                }
                $m \leftarrow \left\lfloor (m \cdot 5^{q_0}) \cdot 2^{-\lambda} \right\rfloor $\;   
                $m \leftarrow m \cdot 2^{( (2^{n_k}-1) + (2^{n_{k-1}}-1) + \cdots + (2^{n_1}-1) + k)\lambda - \sum_{i=k}^1 \sigma_i} $\;
                $s \leftarrow {\sum_{i=k}^1 \left(n_i(\left\lfloor\log_2(5^{q_i})\right\rfloor - p + 1)\right) + \left\lfloor\log_2(5^{q_0})\right\rfloor-p+1}$\;
                result $\leftarrow m \cdot 2^{s}$\;
                \caption{Final multiplication step}
              \end{algorithm}

The whole algorithm schema is presented on Fig.~\ref{fig:algo2}.
Depending on the range of $B$ one can represent it in different
manner, but for our conversion tasks the ranges for $B$ were not that
large, so the numbers $n_j$ were not more than 10 and the loops for
squarings can be easily unrolled. For instance, for the conversions
from binary32, binary64, decimal32 and decimal64 one can use the
expansion of $B$ of the following form:
$$B = 2^8 \cdot q_2 + 2^4 \cdot q_1 + q_0$$
\begin{center}
  \begin{figure}[t]
    \scalebox{.9}{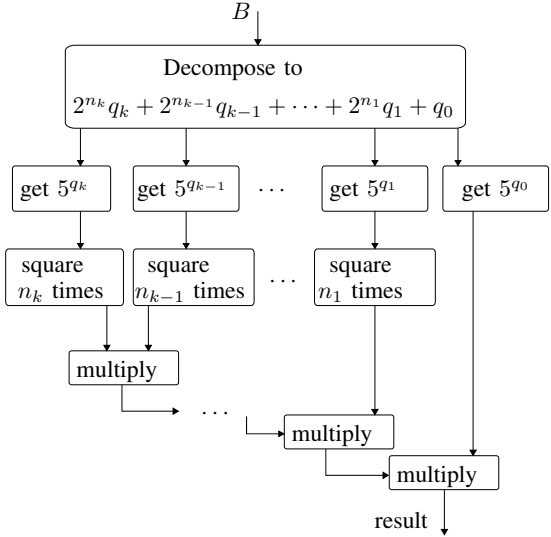}
    \caption{Raising 5 to a huge power}
    \label{fig:algo2}
  \end{figure} 
\end{center}

\section{Error Analysis}
\label{sec:error analysis}

In order to compute the mantissa we use integer arithmetic but on each
squaring/multiplication step we throw away a certain quantity of
bits. So the final error is due to these right shiftings on each
multiplication step. 

We have errors only due to the multiplications, and as we do a lot of
them, we need to define $N$ as the number of all the multiplications
(squaring is just a particular case of multiplication). For each $i$-th factor ($1\le i \le N$)
in~(\ref{eq:5Bform}) we need to perform $n_i$ squarings, thus it gives
us $n_i$ multiplications. In order to get the final result we have to
perform $k$ more multiplications, so the final expression for the $N$
constant is $$N=\sum_{i=1}^kn_i + k.$$ So, the result is a product of
$N$ factors and on each step we have some relative error $\varepsilon_i$. This
means, that if we define $y$ as the exact product without errors, then
what we really compute in our algorithm can be represented as
following:
$$\hat{y} = y \prod_{i=1}^N(1+\varepsilon_i).$$
Thus, the relative error of the computations is 
$$\varepsilon = \frac{\hat{y}}{y}-1=\prod_{i=1}^N(1+\varepsilon_i) - 1$$

Let us prove a lemma that will help us to find the bounds for the
relative error of the result.

\begin{lem}
  Let $N \ge 3$,  $0 \le \bar{\varepsilon} < 1$ and $|\varepsilon_i| \le
  \bar{\varepsilon}$ for all $i \in [1, N]$. Then the following holds:
  $$\left| \prod_{i=1}^{N}(1+\varepsilon_i) - 1 \right | \le (1 + \bar{\varepsilon})^N - 1.$$
\end{lem}

\begin{proof} 
This inequality is equivalent to the following:
$$ -(1 + \bar{\varepsilon})^N + 1 \le  \prod_{i=1}^{N}(1+\varepsilon_i) - 1 \le (1 + \bar{\varepsilon})^N - 1$$

The proof of the right side is trivial.  From the lemma condition we
have $-\bar{\varepsilon} \le \varepsilon_i \le \bar{\varepsilon}$,
which is the same as $1 - \bar{\varepsilon} \le \varepsilon_i +1 \le
\bar{\varepsilon} + 1$ for arbitrary $i$ from the interval
$[1,N]$. Taking into account the borders for $\bar{\varepsilon}$, we
get that $0 < (1+\varepsilon_i) < 2$ for all $i \in [1, N]$. This
means that we can multiply the inequalities $1+\varepsilon_i \le
\bar{\varepsilon}+1$ by $1+\varepsilon_j$ with $j\neq i$. After
performing $N-1$ such multiplications and taking into account that
$1+\varepsilon_i \le \bar{\varepsilon}+1$, we get the following:
$$\prod_{i=1}^N(\varepsilon_i +1) \le (\bar{\varepsilon} +
1)^N.$$

So, the right side is proved.

The same reasoning applies for the left bounds from the lemma
condition, and the family of inequalities $1 - \bar{\varepsilon} \le
\varepsilon_i +1 $ leads to the condition:
$$(1-\bar{\varepsilon})^N - 1 \le \prod_{i=1}^N(1+\varepsilon_i)
- 1.$$ 
So, in order to prove the lemma we have to prove now
that $$-(1 + \bar{\varepsilon})^N + 1 \le(1-\bar{\varepsilon})^N -
1.$$ After regrouping the summands we get the following expression to prove:
$$ 2 \le (1 + \bar{\varepsilon})^N + (1 -  \bar{\varepsilon})^N.$$
Using the binomial coefficients this trasforms to
$$  2 \le 1+ \sum_{i=1}^N {N \choose i} \bar{\varepsilon}^i + 1 + \sum_{i=1}^N {N \choose i} (-\bar{\varepsilon})^i$$
On the right side of this inequality we always have the sum of 2
and some nonnegative terms. So, the lemma is proven.
\end{proof}

The error $\bar{\varepsilon}$ is determined by the basic
multiplication algorithm. It takes two input numbers (each of them is
bounded between $2^{p-1}$ and $2^p$), multiplies them and cuts
$\lambda$ last bits, see line 3 of algorithms~\ref{algo:squarings}
and~\ref{algo:mult}. Thus, instead of $v_j^2$ on each step we get
$v_j^22^{-\lambda} + \delta$, where $-1 < \delta \le 0$. So, the
relative error of the multiplication is bounded by $\left|
  \bar{\varepsilon} \right| \le 2^{-2p+2+\lambda}$.

\section{Implementation Details}
\label{sec:impldet}

While the implementation of the first step is relatively simple, we
need to specify some parameters and techniques that we used to
implement raising 5 to an integer power.

The used computational precision $p$ was equal to 128 bits. The
standard C integer types give us either 32 or 64 bits, so for the
implementation we used the \texttt{uint128\_t} type from GCC that is
realised with two 64-bit numbers. As a shifting parameter $\lambda$ we
took 64, so getting most or least 64 bits out of \texttt{uint128\_t}
number is easy and fast. Squarings and multiplications can be easily
implemented using typecastings and appropriate shifts. Here, for
instance, we put the code of squaring the 64-bit integer. The function
returns two 64-bit integers, so the high and the low word of the
128-bit number.

\lstset{
language=C,
basicstyle=\small\sffamily,
numbers=left,
numberstyle=\tiny,
frame=tb,
columns=fullflexible,
showstringspaces=false,
captionpos=b
}

 \lstloadlanguages{
         C
 }

\begin{lstlisting}[language=C, caption=Example. C code sample for squaring a 64-bit number.]
void square64(uint64_t * rh, 
		     uint64_t * rl, 
		     uint64_t a) {
  uint128_t r;

  r = ((uint128_t) a) * ((uint128_t) a);

  *rl = (uint64_t) r;
  r >>= 64;
  *rh = (uint64_t) r;
}
\end{lstlisting}

The other functions were implemented in the same manner.

We have implemented an run parametrized algorithm for computation of
$5^B$, as the parameter we took the table index size (for entries
$5^{q_i}$) and the working precision $p$. We see
(Fig.~\ref{fig:accuracy}) that the accuracy depends almost linearly on
the precision.

\begin{center}
\begin{figure}
  \scalebox{0.47}{\includegraphics[scale=0.68]{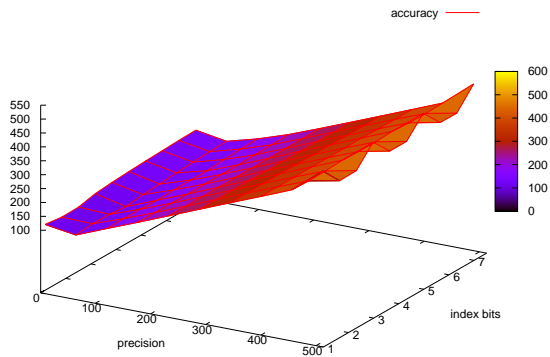}}
\caption{Accuracy as a function of precision and table index size}
\label{fig:accuracy}
\end{figure}
\end{center}

\section{Conclusions}

A novel algorithm for conversion between binary and decimal
floating-point representations has been presented. All the
computations are done in integer arithmetic, so no FP flags or modes
can be influenced. This means that the corresponding code can be made
reentrant.  The exponent determination is exact and can be done with
several basic arithmetic operations, stored constants and a table.
The mantissa computation algorithm uses a small exact table. The
error analysis is given and it corresponds to the experimental
results. The accuracy of the result depends on the computing precision
and the table size.  The conversions are often used and the tables are
multipurpose, so they can be reused by dozens of algorithms. As this
conversion scheme is used everywhere and the tables are not large,
they might be integrated in hardware. The implementation of the
proposed algorithm can be done without loops, so it reduces the
instructions that control the loop, optimizes and therefore
accelerates the code. The described conversion approach was used in
the implementation of the \texttt{scanf} analogue in
\texttt{libieee754} library~\cite{KLSCAN2012}.

\bibliographystyle{IEEEtran}
\bibliography{mybibliography}
\end{document}

%% file: images/algo2Article.pdf_tex
\begingroup%
  \makeatletter%
  \providecommand\color[2][]{%
    \errmessage{(Inkscape) Color is used for the text in Inkscape, but the package 'color.sty' is not loaded}%
    \renewcommand\color[2][]{}%
  }%
  \providecommand\transparent[1]{%
    \errmessage{(Inkscape) Transparency is used (non-zero) for the text in Inkscape, but the package 'transparent.sty' is not loaded}%
    \renewcommand\transparent[1]{}%
  }%
  \providecommand\rotatebox[2]{#2}%
  \ifx\svgwidth\undefined%
    \setlength{\unitlength}{235.27558594bp}%
    \ifx\svgscale\undefined%
      \relax%
    \else%
      \setlength{\unitlength}{\unitlength * \real{\svgscale}}%
    \fi%
  \else%
    \setlength{\unitlength}{\svgwidth}%
  \fi%
  \global\let\svgwidth\undefined%
  \global\let\svgscale\undefined%
  \makeatother%
  \begin{picture}(1,0.96089948)%
    \put(0,0){\includegraphics[width=\unitlength]{algo2Article.pdf}}%
    \put(0.28222222,0.81978836){\color[rgb]{0,0,0}\makebox(0,0)[lb]{\smash{Decompose to}}}%
    \put(0.12095238,0.75259259){\color[rgb]{0,0,0}\makebox(0,0)[lb]{\smash{$2^{n_k}q_k+2^{n_{k-1}}q_{k-1} + \cdots + 2^{n_1}q_1+q_0$
}}}%
    \put(0.02687831,0.60476191){\color[rgb]{0,0,0}\makebox(0,0)[lb]{\smash{get $5^{q_k}$}}}%
    \put(0.23518519,0.60476191){\color[rgb]{0,0,0}\makebox(0,0)[lb]{\smash{get $5^{q_{k-1}}$}}}%
    \put(0.5778836,0.60476191){\color[rgb]{0,0,0}\makebox(0,0)[lb]{\smash{get $5^{q_1}$}}}%
    \put(0.81306879,0.60476191){\color[rgb]{0,0,0}\makebox(0,0)[lb]{\smash{get $5^{q_0}$}}}%
    \put(0.45021164,0.60476191){\color[rgb]{0,0,0}\makebox(0,0)[lb]{\smash{$\cdots$}}}%
    \put(0.02956614,0.47037037){\color[rgb]{0,0,0}\makebox(0,0)[lb]{\smash{square}}}%
    \put(0.02284656,0.42333333){\color[rgb]{0,0,0}\makebox(0,0)[lb]{\smash{$n_k$ times}}}%
    \put(0.254,0.47037037){\color[rgb]{0,0,0}\makebox(0,0)[lb]{\smash{square}}}%
    \put(0.23115344,0.42333333){\color[rgb]{0,0,0}\makebox(0,0)[lb]{\smash{$n_{k-1}$ times}}}%
    \put(0.5778836,0.47037037){\color[rgb]{0,0,0}\makebox(0,0)[lb]{\smash{square}}}%
    \put(0.55772487,0.42333333){\color[rgb]{0,0,0}\makebox(0,0)[lb]{\smash{$n_{1}$ times}}}%
    \put(0.47037037,0.44349206){\color[rgb]{0,0,0}\makebox(0,0)[lb]{\smash{$\cdots$}}}%
    \put(0.12767196,0.28222223){\color[rgb]{0,0,0}\makebox(0,0)[lb]{\smash{multiply}}}%
    \put(0.51068783,0.16798943){\color[rgb]{0,0,0}\makebox(0,0)[lb]{\smash{multiply}}}%
    \put(0.69749207,0.09541799){\color[rgb]{0,0,0}\makebox(0,0)[lb]{\smash{multiply}}}%
    \put(0.34941799,0.20830688){\color[rgb]{0,0,0}\makebox(0,0)[lb]{\smash{$\cdots$}}}%
    \put(0.65851852,0.01343915){\color[rgb]{0,0,0}\makebox(0,0)[lb]{\smash{result}}}%
    \put(0.40317461,0.92058201){\color[rgb]{0,0,0}\makebox(0,0)[lb]{\smash{$B$}}}%
  \end{picture}%
\endgroup%